\documentclass[11pt]{article}
\usepackage[margin=1in]{geometry}
\usepackage[colorlinks=true,
            linkcolor=black,
            citecolor=black,
            urlcolor=black]{hyperref}

\usepackage{setspace}

\usepackage{subcaption}

\usepackage{fixmath}
\usepackage{bm}
\usepackage{amsbsy}
\usepackage{color}
\usepackage{verbatim}
\usepackage{multirow}
\usepackage{amssymb}
\usepackage{amsthm}
\usepackage{amsmath}
\usepackage{array}
\usepackage{mathtools}

\allowdisplaybreaks

\usepackage{graphicx}
\usepackage{tikz}
\usetikzlibrary{positioning}
\usepackage{xcolor}

\usepackage{thm-restate}
\usepackage{algpseudocode}
\usepackage{algorithm}
\usepackage{algorithmicx}

\newtheorem{theorem}{Theorem}
\newtheorem{lemma}{Lemma}

\newtheorem{definition}{Definition}

\newtheorem{remark}{Remark}

\usepackage{acro}
\usepackage{natbib}
\usepackage{booktabs}
\usepackage{makecell}

\newcommand{\size}[1]{\ensuremath{|#1|}}
\newcommand{\lra}[1]{\ensuremath{(#1)}}
\newcommand{\lrc}[1]{\ensuremath{\{#1\}}}
\newcommand{\lrA}[1]{\ensuremath{\left(#1\right)}}
\newcommand{\lrB}[1]{\ensuremath{\left[#1\right]}}
\newcommand{\lrC}[1]{\ensuremath{\left\{#1\right\}}}
\def\A{\mathcal{A}}
\def\F{\mathcal{F}}
\def\O{\mathcal{O}}

\def\T{\mathcal{T}}

\newcommand{\MSF}{\mathrm{msf}}
\newcommand{\pos}[1]{\lrA{#1}_{+}}
\allowdisplaybreaks

\newtheorem{property}{Property}
\def\OPT{\mathrm{opt}}
\def\Odd{\mathrm{Odd}}
\DeclareAcronym{FPT}{
short = FPT,
long = fixed-parameter tractable
}
\DeclareAcronym{LP}{
short = LP,
long = linear programming
}
\DeclareAcronym{MDTSP}{
short = MD-TSP,
alt = multiple-depot TSP,
long = multiple-depot traveling salesman problem
}
\DeclareAcronym{MST}{
short = MST,
long = minimum-cost spanning tree
}
\DeclareAcronym{RPP}{
short = RPP,
long = rural postman problem
}
\DeclareAcronym{RSF}{
short = RSF,
long = rooted spanning forest
}
\DeclareAcronym{TSP}{
short = TSP,
long = traveling salesman problem
}
\DeclareAcronym{XP}{
short = XP,
long = slice-wise polynomial-time
}

\title{A $(1+1/\sqrt{2})$-Approximation for the Multiple-Depot Traveling Salesman Problem}
\author{Jingyang Zhao$^{1}$, Yuxi Liu$^{2}$, Mingyu Xiao$^{2}$\thanks{Corresponding author.} \\
$^{1}$Kyung Hee University \\
$^{2}$University of Electronic Science and Technology of China
}
\date{}
\begin{document}
\maketitle
\begin{abstract}
The metric traveling salesman problem (TSP) is a fundamental problem in combinatorial optimization that asks for a minimum-cost tour covering all clients in a metric graph.
The metric multiple-depot TSP (MD-TSP) is a natural extension, where the graph contains depots and clients, and the objective is to compute a minimum-cost set of tours covering all clients, with each tour starting and ending at the same depot.
When the number of depots is part of the input, an adaptation of the Christofides--Serdyukov heuristic yields an approximation ratio of $2$.
In this paper, we introduce a $(1+1/\sqrt{2})$-approximation algorithm. 
Like the Christofides--Serdyukov heuristic, our algorithm first computes a rooted spanning forest (RSF), then a matching to correct its odd degrees, and finally obtains a solution by shortcutting.
However, instead of using a minimum-cost RSF, we construct an RSF by a primal-dual algorithm for a natural cut relaxation.
The algorithm grows rootless components and the component containing all depots at different rates, adding an edge when its dual constraint becomes tight.
Vertex labels record the times at which clients first become connected to a depot.
The two-speed growth provides a joint bound on the forest cost and two label-dependent terms that also arise in bounding the parity-correction cost.
Balancing the coefficients of these two terms by setting both to $\sqrt{2}-1$ yields the claimed approximation ratio.
\end{abstract}

\section{Introduction}\label{sec:introduction}
The metric \ac{TSP} is one of the most fundamental problems in combinatorial optimization~\citep{GA:06:BOOK}.
It asks for a minimum-cost tour visiting all vertices of a graph.
The problem has been extensively studied since the classical work of \citet{dantzig1954solution} and arises in applications ranging from transportation planning to the ordering of drilling operations on circuit boards~\citep{ABC:2006:book}.
In these applications, vertices represent locations or tasks, and edge costs represent the travel needed between them.

The metric \aca{MDTSP} (\ac{MDTSP}) extends this problem to multiple depots. 
An instance is a complete graph $G=(V,E,c)$, where $V=J\cup D$ and $J\cap D=\emptyset$.
Here, $J=\{v_1,\ldots,v_n\}$ is the set of clients, $D=\{u_1,\ldots,u_k\}$ is the set of depots with $k\geq1$, and $c:V\times V\to\mathbb{R}_{\geq0}$ is a metric edge cost function satisfying $c(x,x)=0$, $c(x,y)=c(y,x)$, and the triangle inequality $c(x,z)\leq c(x,y)+c(y,z)$ for all $x,y,z\in V$.
The goal is to find a minimum-cost set of tours covering all clients, with each tour starting and ending at the same depot.
Depots may remain unused.
The multiple-depot setting captures situations in which vehicles or service teams start from different locations and must collectively serve a set of clients.
For example, in inspection planning, depots can represent vehicle bases and clients can represent inspection sites.
Related multiple-depot vehicle routing models arise in task allocation for unmanned aerial vehicles~\citep{rathinam2007resource,JJSHN:15:CIE}.

The metric \ac{TSP} is APX-hard even when all edge costs between distinct vertices are either one or two~\citep{papadimitriou1993distances}.
Since the metric \ac{MDTSP} contains the metric \ac{TSP} as the case $k=1$, it is also APX-hard.
This motivates approximation algorithms that run in polynomial time and provide provable bounds on solution cost.
Such algorithms have been extensively studied for metric \acp{TSP}~\citep{TV:2024:book,SJA:25:AOR}.

Unless otherwise specified, in the following, we consider problems in the metric setting.

For the \ac{TSP}, the double-tree heuristic~\citep{DRP:77:SICOMP} yields an approximation ratio of $2$.
\citet{christofides1976worst} and \citet{serdyukov1978some} proposed the classical $3/2$-approximation algorithm, which first computes a \ac{MST}, then a minimum-cost perfect matching on its odd-degree vertices, and finally obtains a solution by shortcutting.
More recently, \citet{KKG:21:STOC,DBLP:conf/ipco/KarlinKG23} improved the approximation ratio to $3/2-\varepsilon$ for some constant $\varepsilon>10^{-36}$.

For the \ac{MDTSP}, adapting the double-tree heuristic yields a $2$-approximation algorithm~\citep{rathinam2007resource}.
By extending the Christofides--Serdyukov heuristic, \citet{xu2011analysis} improved this ratio to $2-1/k$ for $k\geq2$, with running time $\O(\size{V}^3)$.
Their algorithm replaces the \ac{MST} with a minimum-cost \ac{RSF}, a spanning forest with one depot in each tree, and then applies parity correction and shortcutting.
They also showed that the ratio $2-1/k$ is tight for this algorithm.
Moreover, for a generalization in which only $s$ of the $k$ available depots may be selected to serve clients, \citet{DBLP:journals/eor/XuR17} proposed a $(2-1/(2s))$-approximation algorithm.

Although the best-known approximation ratio remains $2$ in the general setting, better guarantees can be achieved by allowing the running time to depend on the number of depots $k$.
Using edge exchanges, \citet{xu20153} proposed a $3/2$-approximation algorithm with running time $\O(\size{V}^{k+2}\size{E}^{k-1})$.
This running time is polynomial for any fixed $k$, so the algorithm belongs to \ac{XP} when parameterized by the number of depots~\citep{CFK:2015:book}.
Later, \citet{DBLP:journals/siamcomp/TraubVZ22} proposed the $\Phi$-\ac{TSP} framework for a class of \ac{TSP} variants.
For the \ac{MDTSP}, it yields an $(\alpha+\varepsilon)$-approximation algorithm with running time $\size{V}^{\O_\varepsilon(k)}$ for any constant $\varepsilon>0$, where $\alpha<3/2$ is the best-known ratio for the \ac{TSP}.
Subsequently, \citet{deppert20233} proposed a randomized \ac{FPT} $(3/2+\varepsilon)$-approximation algorithm with running time $(1/\varepsilon)^{\O(k\log k)}\size{V}^{\O(1)}$.
Their algorithm reduces the \ac{MDTSP} to the \ac{RPP} and applies parameterized algorithms for the \ac{RPP}~\citep{DBLP:journals/jcss/GutinWY17}.
Recently, \citet{DBLP:journals/corr/abs-2601-01841} obtained an \ac{FPT} $(\alpha+\varepsilon)$-approximation algorithm with running time $2^{\O(k/\varepsilon)}\size{V}^{\O(1)}$.
These results provide stronger ratios when the number of depots is small.
However, it remained open whether a polynomial-time algorithm could achieve a ratio better than $2$ when $k$ is part of the input.

\subsection{Our Results}
We propose an $\O(\size{V}^3)$-time $(1+1/\sqrt2)$-approximation algorithm for the \ac{MDTSP}, breaking the 2-approximation barrier.

Like the Christofides--Serdyukov-style $(2-1/k)$-approximation algorithm of \citet{xu2011analysis}, our algorithm first constructs an \ac{RSF} $\F$, then corrects its odd degrees by a minimum-cost perfect matching $M$, and finally obtains a solution to the \ac{MDTSP} by shortcutting $E(\F)\uplus M$.
The main difference lies in how we construct the forest and bound its cost together with the matching cost.

Fix $\lambda\in[0,1]$.
Following the primal-dual approach of \citet{DBLP:journals/siamcomp/GoemansW95}, we construct $\F$ by growing dual variables of a natural cut relaxation.
At time $s=0$, all depots are treated as a single root component.
Each rootless component grows at rate $1/2$, while the root component grows at rate $\eta/2$, where $\eta=\frac{1-\lambda}{1+\lambda}$.
An edge is selected when its dual constraint becomes tight under the original edge costs $c$.
When a component first becomes connected to a depot, all its vertices receive the current time as their label $\ell(v)$.
The labels also yield a weighted-forest interpretation of the construction.
Under $w=c+\lambda\theta_\ell$, where $\theta_\ell(uv)=\max\{c(uv)-\min\{\ell(u),\ell(v)\},0\}$, the forest $\F$ is a minimum-weight \ac{RSF}, and each selected edge has weight equal to its event time.

The purpose of the two rates is to bound the forest and matching costs together.
Consider an optimal solution $\T^*=\{T_1^*,\ldots,T_k^*\}$ of cost $\OPT$, whose tours partition $V$ into sets $V_i=V(T_i^*)$.
The growth of rootless components accounts for $c(\F)+\lambda\theta_\ell(\F)$, while that of the root component accounts for $\eta\sum_{i=1}^k\max_{v\in V_i}\ell(v)$.
Using dual feasibility, we obtain $c(\F)+\lambda\theta_\ell(\F)+\eta\sum_{i=1}^k\max_{v\in V_i}\ell(v)\leq\OPT$.

To bound the matching cost, we retain some forest edges joining different sets $V_i$ and then correct the degree parities within each $V_i$.
This gives $c(M)\leq\frac12\OPT+\sum_{i=1}^k\max_{v\in V_i}\ell(v)+\theta_\ell(\F)$.
The two additional terms in this bound have coefficients $\eta$ and $\lambda$, respectively, in the dual budget.
Balancing these coefficients by setting $\lambda=\eta=\sqrt2-1$ yields $c(\F)+c(M)\leq(3/2+\sqrt2)\OPT-\sqrt2c(\F)$.
Together with the bound $c(\F)+c(M)\leq2c(\F)$, this proves the approximation ratio $1+1/\sqrt2$.

\section{Preliminaries}\label{sec:preliminaries}
As defined earlier, an instance of the \ac{MDTSP} is denoted by a metric graph $G=(V,E,c)$.

A \emph{walk} is a sequence of vertices in which consecutive vertices are joined by an edge.
For a walk or multigraph $H$, let $V(H)$, $E(H)$, and $c(H)$ denote its vertex set, edge multiset, and total edge cost, respectively.
Edge sums and cut intersections involving a multigraph count edge multiplicities.
For $S\subseteq V$, let $\delta(S)$ denote the set of edges with exactly one endpoint in $S$, and write $\delta(v)=\delta(\{v\})$.
A \emph{tour} is a closed walk beginning and ending at a depot, and a \emph{trivial} tour consists of a single depot and has cost zero.
By the triangle inequality, repeated clients and intermediate depots can be removed by \emph{shortcutting} without increasing its cost.
Thus, after shortcutting, a tour serving a single client uses two copies of the edge to its depot.

\begin{definition}[The \ac{MDTSP}]
Given $G=(V,E,c)$, the goal is to find tours $\T=\{T_u\}_{u\in D}$ covering all clients, with each $T_u$ starting and ending at $u$, that minimize $c(\T)=\sum_{u\in D}c(T_u)$.
\end{definition}

A feasible solution can be represented as a spanning even-degree multigraph with at least one depot in each component.
Let $\T^*$ be an optimal solution with cost $\OPT=c(\T^*)$.
By assigning each client to one visiting tour and shortcutting all other visits, we may assume that its tours $T_1^*,\ldots,T_k^*$ are vertex-disjoint, each containing exactly one depot.
Thus, the sets $V_i=V(T_i^*)$ partition $V$.

For a multigraph $H$, let $\Odd(H)$ denote its set of odd-degree vertices.
For an even-cardinality vertex set $U$, a \emph{$U$-join} is an edge multiset whose odd-degree vertex set is exactly $U$.
An \ac{RSF} is a spanning forest with exactly one depot in each component, allowing isolated depots.
For nonnegative edge weights $w$, let $\MSF_w(D)$ denote the minimum weight of such a forest.
Let $\pos{x}=\max\lrc{x,0}$. Then, for vertex labels $\ell:V\to\mathbb{R}_{\geq0}$ satisfying $\ell(u)=0$ for all $u\in D$, define
\begin{equation}\label{eqtheta}
\theta_\ell(uv)=\pos{c(uv)-\min\lrc{\ell(u),\ell(v)}}.
\end{equation}

\section{A Primal-Dual Algorithm for the \ac{MDTSP}}\label{sec:algorithm}
The algorithm first constructs an \ac{RSF} $\F$ using a primal-dual approach with a parameter $\lambda\in[0,1]$, then corrects the odd degrees of $\F$ by computing a minimum-cost perfect matching $M$ on $\Odd(\F)$, and finally obtains a solution by shortcutting $E(\F)\uplus M$.
The details are provided in Algorithm~\ref{alg1}.
By the triangle inequality,
\begin{equation}\label{eqnotations}
c(\T)\leq c(\F)+c(M).
\end{equation}

\begin{algorithm}[!t]
\small
\caption{A primal-dual approximation algorithm for the \ac{MDTSP}}
\label{alg1}
\textbf{Input:} An \ac{MDTSP} instance $G=(V,E,c)$ and a parameter $\lambda\in[0,1]$.\\
\textbf{Output:} A feasible solution $\T$.
\begin{algorithmic}[1]
\State Initialize $\T\coloneq\emptyset$.
\State Apply the primal-dual construction in Section~\ref{sec:sweep} to obtain $\F$ and $\ell$.\label{alg1l2}
\State Compute a minimum-cost perfect matching $M$ on $\Odd(\F)$ under $c$.\label{alg1l3}
\For{each nontrivial component $H$ of $E(\F)\uplus M$}\label{alg1l4}
\State Choose a depot $u\in D\cap V(H)$ and compute an Eulerian tour of $H$ starting at $u$.\label{alg1l5}
\State Shortcut repeated clients and intermediate depots to obtain $T_u$, and add $T_u$ to $\T$.\label{alg1l6}
\EndFor\label{alg1l7}
\State Add a trivial tour for each unused depot, and \Return $\T$.\label{alg1l8}
\end{algorithmic}
\end{algorithm}

For each edge $e\in E$, we introduce a variable $x_e\geq 0$, and define $x(E')=\sum_{e\in E'}x_e$ for $E'\subseteq E$. We consider the following \ac{LP} relaxation of the \ac{MDTSP}:
\begin{equation}\label{eqprimal}
\min \lrC{
\sum_{e\in E}c(e)x_e \mid
x(\delta(S))\geq 2 \quad
\forall\,\emptyset\neq S\subseteq J,\quad x_e\geq 0
}.
\end{equation}
Every feasible solution to the \ac{MDTSP} yields a feasible solution to \eqref{eqprimal} of the same cost. 
Hence, the optimal value of \eqref{eqprimal} is a lower bound on $\OPT$. Its dual is
\begin{equation}\label{eqdual}
\max\lrC{
2\sum_{\emptyset\neq S\subseteq J}y_S \mid
\sum_{\substack{\emptyset\neq S\subseteq J\\ e\in\delta(S)}}y_S
\leq c(e) \quad \forall\,e\in E,\quad y_S\geq 0
}.
\end{equation}
With respect to a dual solution, an edge is called \emph{tight} if its constraint in \eqref{eqdual} holds with equality.
By replacing the right-hand side $2$ by $1$ in \eqref{eqprimal}, we obtain a relaxation for rooted connectivity, whose dual has objective $\sum_Sy_S$.
The forest construction satisfies these connectivity requirements.

\begin{remark}
Our approximation ratio is not measured with respect to the value of \eqref{eqprimal}. 
As shown by \citet{deppert20233}, the subtour-elimination relaxation can have an integrality gap arbitrarily close to $2$.
This relaxation additionally imposes the degree constraints $x(\delta(v))=2$ for every $v\in J$, which is stronger than our relaxation.
Thus, the integrality gap of our relaxation is also at least $2$.
\end{remark}

\subsection{The Two-Speed Primal-Dual Construction}\label{sec:sweep}
We now introduce the forest construction.
Fix $\lambda\in[0,1]$, and define $\eta=\frac{1-\lambda}{1+\lambda}$.
Following the primal-dual approach of \citet{DBLP:journals/siamcomp/GoemansW95}, we describe the construction using a continuous time parameter $s\geq0$. We grow variables associated with current components until an edge becomes tight, then select that edge and merge the components.
The main difference is that we allow the root component to grow at a rate no greater than that of the rootless components.

At time $s=0$, let $\F=(V,\emptyset)$, set all dual variables to zero, and treat all depots as belonging to a single root component $R_0$, without adding edges between them.
Give each client its own rootless component, and set $\ell(u)=0$ for every depot $u\in D$; client labels are initially unassigned.
Throughout the construction, $R_0$ denotes the component containing all depots.

At time $s$, let $\A_s$ be the family of current rootless components, and let $A_s=\bigcup_{X\in\A_s}X=V\setminus R_0$.
While $A_s\neq\emptyset$, the dual variables grow at the following rates:
\begin{equation}\label{eqgrowth}
\frac{d y_S}{d s}
=\frac12\mathbf{1}_{\{S\in\A_s\}}
+\frac\eta2\mathbf{1}_{\{S=A_s\}},
\qquad \forall\ \emptyset\neq S\subseteq J,
\end{equation}
where $\mathbf{1}_{\{\cdot\}}$ is the indicator of the condition.
The first term in \eqref{eqgrowth} grows each rootless component at rate $1/2$.
Since $\delta(A_s)=\delta(R_0)$, the second term represents growth of the root component at rate $\eta/2$.
If only one rootless component $X$ remains, then $A_s=X$ and $\frac{d y_X}{d s}=\frac{1+\eta}{2}$.

Continue the growth until an edge constraint joining distinct current components first becomes tight.
Choose any such edge, add it to $\F$, and merge its components without changing the current dual variables.
This edge selection and component merge is an \emph{event}, and the current time is its \emph{event time}.
If a rootless component merges with $R_0$, assign the current time as the label of each vertex in that component; otherwise, assign no labels.
Before resuming the growth, process any remaining tight edges joining distinct current components in the same way, recomputing the components after each merge.
In particular, zero-cost edges may be selected at time zero, and several events may occur at the same time.
The construction terminates when no rootless component remains.

\begin{lemma}\label{lem:dual-growth}
The construction maintains feasibility of \eqref{eqdual}, and every selected edge remains tight.
After exactly $\size{V}-k$ merges, the selected edges form an \ac{RSF} $\F$.
\end{lemma}
\begin{proof}
By definition, growth rates are nonnegative, and growth stops as soon as an edge joining distinct components becomes tight, preserving dual feasibility. After a merge, edges within the resulting component receive no further dual contribution, so selected edges remain tight.

Each selected edge joins two distinct components and reduces the number of rootless components by one.
Since the construction starts with $\size{V}-k$ rootless components, it performs $\size{V}-k$ merges.
After contracting all depots into a single vertex, the selected edges form a tree. 
Before contraction, they form a spanning forest with one depot in each component. Thus, $\F$ is an \ac{RSF}.
\end{proof}

\emph{The event implementation.}
Instead of explicitly maintaining the dual variables, we can compute the next tight edge from its event time.
Consider an edge $uv$ joining distinct current components at time $s$.
If both components are rootless, their endpoints have remained in distinct rootless components during $[0,s]$.
Their component variables contribute $s/2+s/2=s$ to the edge constraint.
Hence, the edge becomes tight at time $c(uv)$.

Suppose that $u$ joined $R_0$ at time $a=\ell(u)$, while $v$ is still rootless.
Before time $a$, the endpoints lie in distinct rootless components,
each contributing $a/2$ to the edge constraint.
After time $a$, only the rootless component containing $v$ and the set $A_s$ contribute, at rates $1/2$ and $\eta/2$, respectively.
Thus, the total contribution at time $s$ is $a+\frac{1+\eta}{2}(s-a)=\frac{(1+\eta)s+(1-\eta)\ell(u)}{2}$.
Let this expression equal $c(uv)$, and then we obtain the time at which $uv$ becomes tight:
\[
s=\frac{2c(uv)-(1-\eta)\ell(u)}{1+\eta}
=(1+\lambda)c(uv)-\lambda\ell(u).
\]
Thus, for an edge joining distinct current components, define its current key as
\begin{equation}\label{eqevent}
 \kappa(uv)=
 \begin{cases}
 c(uv),&u,v\notin R_0,\\
 (1+\lambda)c(uv)-\lambda \ell(u),&u\in R_0,\ v\notin R_0,\\
 (1+\lambda)c(uv)-\lambda \ell(v),&v\in R_0,\ u\notin R_0.
 \end{cases}
\end{equation}
Choose a minimum-key edge $uv$, breaking ties deterministically, and set the event time to $s=\kappa(uv)$.
Merge the two components, assign label $s$ to newly rooted vertices, and update the remaining keys using \eqref{eqevent}.
With the same tie-breaking rule, this event implementation selects exactly the same edges and labels as the dual-growth procedure.

\begin{property}\label{labelpro}
The label $\ell(v)$ records when client $v$ first becomes connected to a depot.
\end{property}

When $\lambda=0$, we have $\eta=1$, all components grow at the same rate, and the keys equal the costs; the construction simulates Kruskal's algorithm.
When $\lambda=1$, we have $\eta=0$, and only rootless components grow.

\begin{figure}[t]
\centering
\caption{The primal-dual construction for a line metric, with $\lambda=1/2$, $\eta=1/3$, and $h=9/2$.
Squares are depots; shaded vertices belong to the root component $R_0$.
In (b)--(d), thick edges are selected and carry their event times; dashed edges are unselected and carry their current keys $\kappa$.}
\label{fig:event-example}
\begin{tikzpicture}[
    x=1.35cm,y=1cm,
    client/.style={circle,draw,fill=white,minimum size=6mm,inner sep=0pt},
    depot/.style={rectangle,draw,fill=black!18,minimum size=6mm,inner sep=0pt},
    rooted/.style={fill=black!18},
    pending/.style={draw=black!40,dashed},
    chosen/.style={draw=black,line width=1.2pt},
    every node/.style={font=\small}
    ]

\begin{scope}
    \node[anchor=west] at (-.3,.8) {(a) Input: edge labels are costs $c$.};
    \node[depot] (r1) at (0,0) {$u_1$};
    \node[client] (a1) at (1,0) {$a$};
    \node[client] (b1) at (2,0) {$b$};
    \node[client] (d1) at (3,0) {$d$};
    \node[depot] (r2) at (4,0) {$u_2$};

    \draw[pending] (r1)--node[above] {$2$} (a1);
    \draw[pending] (a1)--node[above] {$1$} (b1);
    \draw[pending] (b1)--node[above] {$4$} (d1);
    \draw[pending] (d1)--node[above] {$6$} (r2);

    \node[anchor=west] at (-.3,-.65) {$\ell(a),\ell(b),\ell(d)$ unassigned.};
\end{scope}

\begin{scope}[xshift=7cm]
    \node[anchor=west] at (-.3,.8) {(b) Event $s=1$: select $ab$.};
    \node[depot] (r1) at (0,0) {$u_1$};
    \node[client] (a1) at (1,0) {$a$};
    \node[client] (b1) at (2,0) {$b$};
    \node[client] (d1) at (3,0) {$d$};
    \node[depot] (r2) at (4,0) {$u_2$};

    \draw[pending] (r1)--node[above] {$3$} (a1);
    \draw[chosen] (a1)--node[above] {$1$} (b1);
    \draw[pending] (b1)--node[above] {$4$} (d1);
    \draw[pending] (d1)--node[above] {$9$} (r2);
    \node[anchor=west] at (-.3,-.65) {$\{a,b\}$ is still rootless; no labels assigned.};
\end{scope}

\begin{scope}[yshift=-2.5cm]
    \node[anchor=west] at (-.3,.8) {(c) Event $s=3$: select $u_1a$.};
    \node[depot] (r1) at (0,0) {$u_1$};
    \node[client,rooted] (a1) at (1,0) {$a$};
    \node[client,rooted] (b1) at (2,0) {$b$};
    \node[client] (d1) at (3,0) {$d$};
    \node[depot] (r2) at (4,0) {$u_2$};

    \draw[chosen] (r1)--node[above] {$3$} (a1);
    \draw[chosen] (a1)--node[above] {$1$} (b1);
    \draw[pending] (b1)--node[above] {$h$} (d1);
    \draw[pending] (d1)--node[above] {$9$} (r2);

    \node[anchor=west] at (-.3,-.65) {$\ell(a)=\ell(b)=3$; $\kappa(bd):4\longrightarrow h$.};
\end{scope}

\begin{scope}[xshift=7cm,yshift=-2.5cm]
    \node[anchor=west] at (-.3,.8) {(d) Event $s=h$: select $bd$; final forest.};
    \node[depot] (r1) at (0,0) {$u_1$};
    \node[client,rooted] (a1) at (1,0) {$a$};
    \node[client,rooted] (b1) at (2,0) {$b$};
    \node[client,rooted] (d1) at (3,0) {$d$};
    \node[depot] (r2) at (4,0) {$u_2$};

    \draw[chosen] (r1)--node[above] {$3$} (a1);
    \draw[chosen] (a1)--node[above] {$1$} (b1);
    \draw[chosen] (b1)--node[above] {$h$} (d1);

    \node[anchor=west] at (-.3,-.65) {$\ell(d)=h=\max\{h,1,3\}$.};
\end{scope}
\end{tikzpicture}
\end{figure}
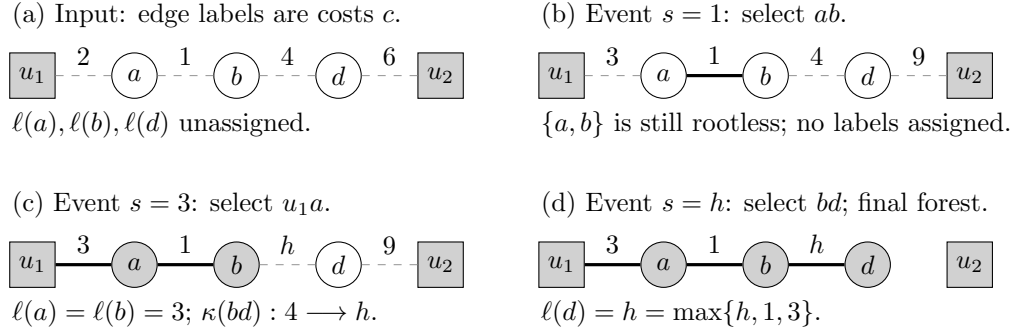

\emph{An example.}
Take $\lambda=1/2$, two depots $u_1,u_2$, and three clients $a,b,d$ at positions $0,13,2,3,7$, respectively, on the real line, as shown in Figure~\ref{fig:event-example}.
The cost of every pair is the absolute difference of its positions; thus the graph is a metric graph.
Initially, the minimum key is $\kappa(ab)=c(ab)=1$.
Selecting $ab$ merges two rootless components, so neither client receives a label.
The next selected edge is $u_1a$, whose key is $3$.
It merges the component $\{a,b\}$ with $R_0$ and assigns both clients the label $3$.
The key of $bd$ then increases from $4$ to $h:=\frac32\cdot4-\frac12\cdot3=\frac92$.
The other edges that can merge the component containing $d$ with $R_0$ have keys
$\kappa(ad)=6$, $\kappa(u_1d)=21/2$, and $\kappa(u_2d)=9$.
Hence, $bd$ is selected last, and $\ell(d)=h$.
The event times are $1<3<h$.

\subsection{Properties and Running Time}\label{sec:forest-properties}
The forest construction defines the labels and selected edges.
We define an edge weight function 
\begin{equation}\label{eqfixed}
w = c + \lambda \theta_\ell.
\end{equation}
Although $w$ may not satisfy the triangle inequality, we have the following properties.

\begin{lemma}\label{lem:sweep}
For $\lambda\in[0,1]$, the forest construction computes $\F$ and $\ell$ in $\O(\size{V}^2\log\size{V})$ time such that $w(\F)=\MSF_w(D)$.
Moreover, each edge $e\in E(\F)$ has final weight $w(e)$ equal to its event time.
\end{lemma}
\begin{proof}
Event times are nondecreasing by the dual-growth construction.
We first show that each selected edge has final weight equal to its event time.
We then prove the minimum-weight property and bound the running time.

\emph{(i) The final weights of selected edges.}
We use the labeling $\ell$ and the weight function in \eqref{eqfixed}.
If an edge $e=uv$ joins two rootless components at time $s$, then $s=c(e)$ and $\ell(u),\ell(v)\geq s$.
Thus, $\theta_\ell(e)=0$ and $w(e)=s$.
Next, suppose $e=uv$ merges a rootless component with $R_0$ through an already labeled vertex $u$.
The event time is $s=(1+\lambda)c(e)-\lambda \ell(u)$.
By Property~\ref{labelpro}, $s\geq \ell(u)$, and hence $c(e)\geq \ell(u)$.
Since $\ell(v)=s$, it follows that
\[
w(e)=c(e)+\lambda(c(e)-\ell(u))=(1+\lambda)c(e)-\lambda \ell(u)=s.
\]
Therefore, every selected edge has final weight equal to its selection time.

\emph{(ii) Optimality of the forest under $w$.}
At any stage, the current key of every candidate edge is a lower bound on its final weight.
For an edge between two rootless components, this follows from $w(e)=c(e)+\lambda\theta_\ell(e)\geq c(e)$ since its current key is $c(e)$.
For an edge from a labeled vertex $u$ to a rootless vertex $v$, Property~\ref{labelpro} implies $\ell(v)\geq \ell(u)$, so
\[
w(uv)=c(uv)+\lambda\pos{c(uv)-\ell(u)}\geq(1+\lambda)c(uv)-\lambda \ell(u)=\kappa(uv).
\]

As shown in (i), the selected edge $e$ has final weight equal to
its current key. Since $e$ minimizes the current key, every
candidate edge $f$ satisfies
\[
w(e)=\kappa(e)\leq\kappa(f)\leq w(f).
\]
Hence, $e$ is a minimum-weight edge among all edges joining
distinct current components, so each selection is a valid
Kruskal step under the final weights $w$ after treating
all depots as a single vertex.
After $|V|-k$ merges, the selected edges form an \ac{RSF} $\F$ with $w(\F)=\MSF_w(D)$.

\emph{(iii) Running time.}
The construction performs $\size{V}-k$ merges. A simple implementation scans all edges after each merge, which takes $\O(\size{V}^3)$ time.
Using the priority-queue approach of \citet[Section~2.3]{DBLP:journals/siamcomp/GoemansW95}, this can be reduced to $\O(\size{V}^2\log\size{V})$. Maintain all candidate edges in a min-heap and a union--find structure. Candidate keys never decrease, and the key of an edge can increase at most once, when one endpoint first joins $R_0$. Hence, keys can be updated lazily upon extraction: discard an edge whose endpoints are already in the same component, reinsert it if its key has increased, and otherwise select it. Since stored keys are lower bounds on their current values, the selected edge always has minimum current key. Thus, each edge incurs only $\O(1)$ heap operations and component queries, for a total of $\O\lra{(\size{E}+\size{V})\log\size{V}}=\O(\size{V}^2\log\size{V})$ time.
\end{proof}

\section{The Analysis}\label{sec:analysis}
In this section, we prove the approximation ratio of Algorithm~\ref{alg1} by comparing its cost with that of the optimal solution $\T^*$.
We first bound the matching cost using the vertex labels and the optimal tours $T_i^*$.
We then use the two-speed dual growth to bound the forest cost and the same label terms.
Finally, we balance the two growth rates and combine these bounds.

\subsection{Bounding $c(M)$}
\begin{lemma}\label{lem:parity}
It holds that $c(M)\le \frac12\OPT+\sum_i\max_{v\in V_i}\ell(v)+\theta_\ell(\F)$.
\end{lemma}
\begin{proof}
We construct an $\Odd(\F)$-join by retaining edges of $\F$ between the sets $V_i$ and correcting parity within each $V_i$ using $T_i^*$.

\emph{(i) Edges between the sets $V_i$.}
Contract each $V_i$ in $\F$ into a single supervertex, and repeatedly delete loops and cycles, including two-edge cycles formed by parallel edges. These deletions preserve the degree parity of every supervertex and leave a forest. Let $Y$ be the spanning graph on $V$ consisting of the retained edges. Root each tree of the contracted forest arbitrarily. For a retained edge $uv$ directed from its parent part $V_i$ to its child part $V_j$, by the definition of $\theta_\ell$ in \eqref{eqtheta}, 
\[
 c(uv)\leq\min\{\ell(u),\ell(v)\}+\theta_\ell(uv)
 \leq\max_{x\in V_j}\ell(x)+\theta_\ell(uv).
\]
Charge the label-maximum term to the child part $V_j$. Each part is the child of at most one edge, so these terms sum to at most $\sum_{i=1}^k\max_{v\in V_i}\ell(v)$.
Since $Y\subseteq E(\F)$,
\begin{equation}\label{eqcontracted-cost}
 c(Y)\leq \sum_{i=1}^k\max_{v\in V_i}\ell(v)+\sum_{e\in E(Y)}\theta_\ell(e)
 \leq \sum_{i=1}^k\max_{v\in V_i}\ell(v)+\theta_\ell(\F).
\end{equation}

\emph{(ii) Parity correction within each $V_i$.}
Set $W=\Odd(\F)\triangle\Odd(Y)$, where $\triangle$ denotes symmetric difference.
Since the deletions preserve the parity of the number of edges leaving each $V_i$, $\size{W\cap V_i} \equiv \size{\Odd(\F)\cap V_i} +\size{\Odd(Y)\cap V_i} \equiv 0 \pmod 2$.
Since $\size{W\cap V_i}$ is even, by the triangle inequality, there exists a perfect matching $M_i$ on $W\cap V_i$ with $c(M_i)\leq c(T_i^*)/2$~\citep{christofides1976worst}.

Finally, let $H=(V,Y\uplus\biguplus_iM_i)$, which satisfies $\Odd(H)=\Odd(Y)\triangle W=\Odd(\F)$.
By \eqref{eqcontracted-cost},
\[
c(H)\leq \sum_{i=1}^k\max_{v\in V_i}\ell(v)+\theta_\ell(\F)+\frac12\sum_i c(T_i^*)=\sum_{i=1}^k\max_{v\in V_i}\ell(v)+\theta_\ell(\F)+\frac12\OPT.
\]

Since $c$ is a metric, a minimum-cost $\Odd(\F)$-join has the same cost as a minimum-cost perfect matching on $\Odd(\F)$~\citep{lawler2001combinatorial}. Therefore, $c(M)\leq c(H)$, which proves the lemma.
\end{proof}

\subsection{Bounding the Forest Cost and Vertex Labels}

\begin{lemma}\label{lem:budget}
For $\lambda\in[0,1]$ and $\eta=\frac{1-\lambda}{1+\lambda}$, it holds that $c(\F)+\lambda\theta_\ell(\F)+\eta\sum_{i=1}^k\max_{v\in V_i}\ell(v)\leq\OPT$.
\end{lemma}
\begin{proof}
For $s\geq0$, let $G_s=(V,\{e\in E:w(e)\leq s\})$, and let $a(s)$ be the number of its components containing no depot.
By Property~\ref{labelpro}, a vertex $v$ lies in a depot-containing component of $G_s$ if and only if $\ell(v)\leq s$. After treating all depots as a single component, $G_s$ has $a(s)+1$ components.
Note that $\F$ contains exactly $a(s)$ edges of weight greater than $s$. Thus, 
\begin{equation}\label{lem4eq1}
\int_0^\infty a(s)ds = w(\F)=c(\F)+\lambda\theta_\ell(\F).
\end{equation}

Let $\lrc{y_S}$ be the final dual solution, which is feasible by Lemma~\ref{lem:dual-growth}.
For each $e\in E$, let $m_e$ be its multiplicity in $E(\T^*)$.
By \eqref{eqdual},
\begin{equation}\label{lem4eq2}
\OPT=\sum_{e\in E}c(e)m_e\geq\sum_{e\in E}m_e\sum_{\substack{\emptyset\neq S\subseteq J\\e\in\delta(S)}}y_S=\sum_{\emptyset\neq S\subseteq J}y_S\size{E(\T^*)\cap\delta(S)}.
\end{equation}

By the growth rules in \eqref{eqgrowth},
\begin{equation}\label{lem4eq3}
\sum_{\emptyset\neq S\subseteq J}y_S\size{E(\T^*)\cap\delta(S)} = \int_0^\infty\left(
\frac12\sum_{X\in\A_s}\size{E(\T^*)\cap\delta(X)}
+\frac\eta2\size{E(\T^*)\cap\delta(A_s)}
\right)\,ds.
\end{equation}

\emph{(i) Contributions of rootless components.}
For each $X\in\A_s$, choose a client in $X$ and its tour in $\T^*$.
The tour visits that client but starts and ends at a depot outside $X$, so it must enter and leave $X$.
Thus, $\size{E(\T^*)\cap\delta(X)}\geq2$, and $\frac12\sum_{X\in\A_s}\size{E(\T^*)\cap\delta(X)}\geq a(s)$.
By \eqref{lem4eq1}, 
\begin{equation}\label{lem4eq4}
\int_0^\infty
\frac12\sum_{X\in\A_s}\size{E(\T^*)\cap\delta(X)}\,ds\geq c(\F)+\lambda\theta_\ell(\F).
\end{equation}

\emph{(ii) Contributions of the root component.}
By Property~\ref{labelpro}, if $\max_{v\in V_i}\ell(v)>s$, then $T_i^*$ contains a client in $A_s$, while its depot lies outside $A_s$.
Thus, $T_i^*$ crosses $\delta(A_s)$ at least twice, so $\size{E(\T^*)\cap\delta(A_s)}\geq 2\size{\{i:\max_{v\in V_i}\ell(v)>s\}}$. Therefore,
\begin{equation}\label{lem4eq5}
\int_0^\infty
\frac\eta2\size{E(\T^*)\cap\delta(A_s)}\,ds\geq\int_0^\infty
\eta\size{\{i:\max_{v\in V_i}\ell(v)>s\}}\,ds=\eta\sum_{i=1}^k\max_{v\in V_i}\ell(v),
\end{equation}
where the equality follows from $\int_0^\infty \mathbf{1}_{\{a>s\}}\,ds=a$ for $a\geq0$.

By \eqref{lem4eq2}--\eqref{lem4eq5}, we obtain the inequality in the lemma.
\end{proof}

\subsection{Analyzing the Approximation Ratio}
Lemma~\ref{lem:parity} contains two additional terms, $\theta_\ell(\F)$ and $\sum_{i=1}^k\max_{v\in V_i}\ell(v)$. We use Lemma~\ref{lem:budget} to bound their sum and choose the parameter to obtain the following result.

\begin{lemma}\label{lem:branch}
When $\lambda=\sqrt{2}-1$, it holds that $c(\T)\leq c(\F)+c(M)
\leq \lrA{\frac32+\sqrt2}\OPT-\sqrt2c(\F)$.
\end{lemma}
\begin{proof}
By Lemma~\ref{lem:budget}, $\theta_\ell(\F)+\sum_{i=1}^k\max_{v\in V_i}\ell(v)\leq \frac{\OPT-c(\F)}{\min\{\lambda,\eta\}}$.
Then, by \eqref{eqnotations} and Lemma~\ref{lem:parity},
\[
c(\T)\leq c(\F)+c(M)
\leq c(\F)+\frac12\OPT
+\frac{\OPT-c(\F)}{\min\{\lambda,\eta\}}.
\]
Since $\lambda\in[0,1]$ and $\eta=\frac{1-\lambda}{1+\lambda}$, to maximize $\min\{\lambda,\eta\}$, we choose $\lambda=\frac{1-\lambda}{1+\lambda}$, and then $\lambda=\sqrt2-1$.
In this case, we obtain $c(\T)\leq c(\F)+c(M) \leq \lra{\frac32+\sqrt2}\OPT-\sqrt2c(\F)$.
\end{proof}

\begin{theorem}\label{thm:main}
The \ac{MDTSP} admits an $\O(\size{V}^3)$-time $\lra{1+1/\sqrt{2}}$-approximation algorithm.
\end{theorem}
\begin{proof}
We apply Algorithm~\ref{alg1} with $\lambda=\sqrt2-1$.
We first analyze the approximation ratio.
Since $E(\F)$ is an $\Odd(\F)$-join, $c(M)\le c(\F)$, and thus $c(\F)+c(M)\leq2c(\F)$.
Then, by Lemma~\ref{lem:branch},
\begin{align*}
c(\F)+c(M)&\leq\min\lrC{2c(\F),\lra{3/2+\sqrt2}\OPT-\sqrt2c(\F)}\\
&\leq\lra{\sqrt{2}-1}\cdot 2c(\F) + \lra{2-\sqrt{2}}\cdot\lrB{\lra{3/2+\sqrt2}\OPT-\sqrt2c(\F)}=\lra{1+1/\sqrt2}\OPT.
\end{align*}

We now analyze the running time of Algorithm~\ref{alg1}.
Line~\ref{alg1l2} takes $\O(\size{V}^2\log|V|)$ time by Lemma~\ref{lem:sweep}, line~\ref{alg1l3} takes $\O(\size{V}^3)$ time, dominated by computing a minimum-cost perfect matching~\citep{galil1986matching}, and lines~\ref{alg1l4}--\ref{alg1l8} take $\O(\size{V})$ time~\citep{christofides1976worst}.
Thus, the total running time is $\O(\size{V}^3)$.
\end{proof}

\section{Conclusion}\label{sec:conclusion}

In this paper, we proposed an $\O(\size{V}^3)$-time $(1+1/\sqrt{2})$-approximation algorithm for the \ac{MDTSP}, improving the previous approximation ratio of $2$ when the number of depots is part of the input.
Our algorithm uses a two-speed primal-dual approach to construct an \ac{RSF}, together with vertex labels that connect the forest construction to the parity-correction analysis.
By combining the bounds on the forest and matching costs and balancing the two growth rates, we obtain the claimed approximation ratio.
It would be interesting to investigate whether the primal-dual approach developed in this paper can be extended to other multiple-depot routing problems.

\bibliographystyle{apalike} 
\bibliography{main} 
\end{document}